\documentclass[english]{article}               
\usepackage{geometry}
\usepackage{amsmath,amsthm} 
\newtheorem{theorem}{Theorem}[section]
\newtheorem{proposition}[theorem]{Proposition}
\newtheorem{definition}[theorem]{Definition}

\newtheorem{lemma}[theorem]{Lemma}
\usepackage[T1]{fontenc}
\usepackage[latin9]{inputenc}
\usepackage{amsmath}
\usepackage{babel}
\usepackage{amsfonts,amssymb,epsfig,multicol}
\usepackage{rotating}
\usepackage{float}
\usepackage{array}

\newcommand\undermat[2]{%
  \makebox[0pt][l]{$\smash{\underbrace{\phantom{%
    \begin{matrix}#2\end{matrix}}}_{\text{$#1$}}}$}#2}

\begin{document}
\date{}
\title{Some Bounds on Binary LCD Codes}
\author{Lucky Galvez \thanks{ Department of Mathematics,
Sogang University,
Seoul 04107, South Korea.
{Email: \tt  legalvez97@gmail.com}}
\and Jon-Lark Kim\thanks{ Department of Mathematics,
Sogang University,
Seoul 04107, South Korea.
{Email: \tt jlkim@sogang.ac.kr}}
\and Nari Lee\thanks{ Department of Mathematics,
Sogang University,
Seoul 04107, South Korea.
{Email: \tt narilee3@gmail.com}}
\and Young Gun Roe\thanks{ Department of Mathematics,
Sogang University,
Seoul 04107, South Korea.
{Email: \tt ygroe@naver.com}}
\and Byung-Sun Won\thanks{ Department of Mathematics,
Sogang University,
Seoul 04107, South Korea.
{Email: \tt byungsun08@gmail.com}}
}

\maketitle
\begin{abstract}

A linear code with a complementary dual (or LCD code) is defined to be a linear code $C$ whose dual code $C^{\perp}$ satisfies $C \cap C^{\perp}$= $\left\{ \mathbf{0}\right\} $. Let $LCD{[}n,k{]}$ denote  the maximum of possible values of $d$ among $[n,k,d]$ binary LCD codes. We give exact values of $LCD{[}n,k{]}$ for $1 \le k \le n \le 12$.
 We also show that $LCD[n,n-i]=2$ for any $i\geq2$ and $n\geq2^{i}$. Furthermore, we show that $LCD[n,k]\leq LCD[n,k-1]$ for $k$ odd and  $LCD[n,k]\leq LCD[n,k-2]$ for $k$ even.
\end{abstract}\maketitle
\smallskip
\noindent \textbf{Keywords :} binary LCD codes \and bounds\and linear code 


\section{Introduction}

A linear code with complementary dual (or LCD code) was   first   introduced by Massey~\cite{M1}  as a reversible code in 1964.  Afterwards, LCD codes were extensively studied in literature and widely applied in data storage, communications systems, consumer electronics, and cryptography.

In \cite{M2} Massey showed that there exist asymptotically good LCD codes. Esmaeili and Yari~\cite{E1} identified a few classes of LCD quasi-cyclic codes. For bounds of LCD codes, Tzeng and Hartmann~\cite{T1} proved that the minimum distance of a class of reversible codes is greater than that given by the BCH bound. Sendrier~\cite{S1} showed that LCD codes meet the asymptotic Gilbert-Varshamov bound using the hull dimension spectra of linear codes. Recently, Dougherty et al.~\cite{key-1} gave a linear programming bound on the largest size of an LCD$[n,d]$. Constructions of LCD codes were studied by Mutto and Lal~\cite{M3}. Yang and Massey~\cite{Y1} gave a necessary and sufficient condition for a cyclic code to have a complementary dual. It is also shown by Kandasamy et al.~\cite{V1} that maximum rank distance codes generated by the trace-orthogonal-generator matrices are LCD codes.  In 2014, Calet and Guilley~\cite{C1} introduced several constructions of LCD codes and investigated an application of LCD codes against side-channel attacks(SCA). Shortly after,  Mesnager et al.~\cite{M4} provided a construction of algebraic geometry LCD codes which could be good candidates to be resistant against SCA. Recently Ding et al.~\cite{D1} constructed several families of reversible cyclic codes over finite fields.

\medskip

The purpose of this paper is to study exact values of $LCD[n,k]$ (see~\cite{key-1}) which is the maximum of possible values of $d$ among $[n,k,d]$ binary LCD codes.
 We give exact values of  $LCD[n,2]$ in Section  \ref{sec:3}.  In Section~\ref{sec:4}, we investigate $LCD[n,k]$ and show that $LCD[n,n-i]=2$ for any $i\geq2$ and  $n\geq2^{i}$. We prove that  $LCD[n,k]\leq LCD[n,k-1]$ for $k$ odd and that $LCD[n,k]\leq LCD[n,k-2]$ for  $k$ even  using the notion of principal submatrices. In Section \ref{sec:5}, we  give exact values for $LCK[n,d]$.  We have included tables for $LCD\left[n,k\right]$  for $1 \le k \le n \le 12$ and $LCK\left[n,d\right]$ for $1 \le d \le n \le 12$.


\section{$LCD [n,2]$}
\label{sec:3}

Let $GF(q)$ be the finite field with $q$ elements.
\emph{An $[n,k]$ code} $C$ over $GF(q)$ is a $k$-dimensional subspace of $GF(q)^{n}$. If $C$ is a linear code, we let

\[C^{\perp}=\left\{ \mathbf{u\in}V\mid\mathbf{u\cdot\mathbf{w}}=0\: {\mbox{for  all }} \mathbf{w\in} C\right\}. \]

We call $C^{\perp}$ the \emph{dual} or \emph{orthogonal} code of $C$.

\begin{definition}

A {\em linear code with complementary dual} (LCD code) is a linear code $C$ satisfying $C \cap C^{\bot}=\left\{ \mathbf{0}\right\}. $

\end{definition}

We note that if $C$ is an LCD code, then so is $C^{\perp}$ because ($C^{\perp}$)$^{\perp}$= C.  The following proposition will be frequently used in the later sections.

\begin{proposition}$(\cite{M2})$
\label{prop:1}

Let $G$ be a generator matrix for a code over $GF(q)$. Then $\det(GG^{T})\neq0$ if and only if $G$ generates an $\mathrm{LCD}$ code.

\end{proposition}

Throughout the rest of the paper, we consider only binary codes.
Dougherty et al.~\cite{key-1} introduced $LCD[n,k]$ which denotes the maximum of possible values of $d$ among $[n,k,d]$ binary LCD codes. Formally we can define it as follows.

\begin{definition}
$LCD\left[n,k\right]:=\max\left\{ d\mid {\mbox{there exists a binary}} \left[n,k,d\right]\ {\mbox{LCD code}}  \right\}. $
\end{definition}

Dougherty et al.~\cite{key-1} gave a few bounds on $LCD\left[n,k\right]$ and  exact values of $LCD\left[n, k\right]$ for $k=1$ only.

Now we obtain  exact values of $LCD\left[n, k\right]$ for $k=2$ for any $n$.

\begin {lemma}

$LCD[n,2]$$\leq \left\lfloor \frac{2n}{3} \right\rfloor$ for $n \geq 2$.

\end{lemma}

\begin {proof}
By the Griesmer Bound~\cite{key-6}, any binary linear $[n,k,d]$ code satisfies
\[
n \ge \sum_{i=0}^{k-1} \left\lceil \frac{d}{2^i} \right\rceil.
\]
Letting $k=2$, we have $n \ge d +  \frac{d}{2}.$
Hence
\[d \le \left\lfloor \frac{2n}{3} \right\rfloor.\]

Therefore any LCD $[n,k,d]$ code must satisfy this inequality.

\end {proof}

\begin{proposition}\label{prop:2} Let $n \ge 2$. Then
$LCD[n,2] = \left\lfloor \frac{2n}{3} \right\rfloor$ for $n \equiv 1,\pm 2, {\mbox{or}}~ 3 \pmod 6$.

\end{proposition}

\begin{proof}
We only need to show the existence of LCD codes with minimum distance achieving the bound $d = \left\lfloor \frac{2n}{3} \right\rfloor$.\\
\indent (i) Let $n \equiv 1\,\,(\text{mod }6)$, i.e. $n=6m+1$ for some positive integer $m$. Consider the code with generator matrix
$$G= \left[ \begin{array}{c|c|c} 1\ldots 1 & 1 \ldots 1\ & 0 \ldots 0 \\ \undermat{2m+1}{0 \ldots 0} & \undermat{2m-1}{1 \ldots  1} & \undermat{2m+1}{1 \ldots 1} \end{array} \right].$$
\bigbreak
\bigbreak

This code has minimum weight $4m =  \left\lfloor \frac{2(6m+1)}{3} \right\rfloor$ and $GG^T = \left[ \begin{array}{cc} 0 & 1 \\ 1 & 0 \end{array} \right]$, i.e., $\det(GG^T) = 1 \neq 0$. Therefore this code is an LCD code.\\
\indent(ii) Let $n \equiv \pm 2\,\,(\text{mod }6)$, i.e., $n=6m+ 2$ for some non negative integer $m$, or $n=6m- 2$ for some positive integer $m$ . Consider the code with generator matrix
$$G= \left[ \begin{array}{c|c|c} 1\ldots 1 & 1 \ldots 1\ & 0 \ldots 0 \\ \undermat{2m+k}{0 \ldots 0} & \undermat{2m}{1 \ldots  1} & \undermat{2m+k}{1 \ldots 1} \end{array} \right]$$
for $k = 1,-1$.

If $k = 1$, this code has minimum weight $4m+1 =  \left\lfloor \frac{2(6m+2)}{3} \right\rfloor$ and $GG^T = \left[ \begin{array}{cc} 1 & 0 \\ 0 & 1 \end{array} \right]$, i.e., $\det(GG^T) = 1 \neq 0$. Therefore this code is an LCD code.

If $k = -1$, this code has minimum weight $4m-2 =  \left\lfloor \frac{2(6m-2)}{3} \right\rfloor$ and $GG^T = \left[ \begin{array}{cc} 1 & 0 \\ 0 & 1 \end{array} \right]$, i.e., $\det(GG^T) = 1 \neq 0$. Therefore this code is an LCD code.\\
\indent(iii) Let $n \equiv 3\,\,(\text{mod }6)$, i.e., $n=3i$ for some positive odd integer $i$. Consider the code with generator matrix
$$G= \left[ \begin{array}{c|c|c} 1\ldots 1 & 1 \ldots 1\ & 0 \ldots 0 \\ \undermat{i}{0 \ldots 0} & \undermat{i}{1 \ldots  1} & \undermat{i}{1 \ldots 1} \end{array} \right].$$
\bigbreak
This code has minimum weight $2i =  \left\lfloor \frac{2(3i)}{3} \right\rfloor$ and $GG^T = \left[ \begin{array}{cc} 0 & 1 \\ 1 & 0 \end{array} \right]$, i.e., $\det(GG^T) = 1 \neq 0$. Therefore this code is an LCD code.

\end{proof}

\begin{proposition}\label{prop:3} Let $n \ge 2$. Then
$LCD[n,2] = \left\lfloor \frac{2n}{3} \right\rfloor - 1$ for $n \equiv 0,-1 \pmod{6}$.
\end{proposition}

\begin {proof}
\indent(i) Let $n \equiv 0 \pmod{6}$. Consider the generator matrix $G$ in $(iii)$ of the proof of Proposition \ref{prop:2}, this time taking $i$ to be an even integer. If the weight of any row of $G$ is increased by one, the weight of the sum of the two rows is decreased by one. Hence, $G$ is the only generator matrix for a binary code that achieves the upper bound, up to equivalence. Clearly, det$(GG^T) = 0$ and so the code is not LCD. It follows that there are no LCD code with minimum distance $\left\lfloor \frac{2n}{3} \right\rfloor$ for $n\equiv 0 \,\,(\text{mod }6)$.

Next, consider the code with generator matrix
$$G= \left[ \begin{array}{c|c|c} 1\ldots 1 & 1 \ldots 1\ & 0 \ldots 0 \\ \undermat{i+1}{0 \ldots 0} & \undermat{i-1}{1 \ldots  1} & \undermat{i}{1 \ldots 1} \end{array} \right]$$
\\[0.2cm]
This code has minimum weight $2i-1 =  \left\lfloor \frac{2(3i)}{3} \right\rfloor - 1$. We note that $GG^T = \left[ \begin{array}{cc} 0 & 1 \\ 1 & 1 \end{array} \right]$, i.e., det$(GG^T) = 1 \neq 0$. Therefore this code is an LCD code.

\indent(ii) Let $C$ be a binary code of length $n \equiv -1\,\,(\text{mod }6)$, i.e., $n=3i-1$ for some positive even $i$. Without loss of generality, the generator matrix for $C$ can be expressed in the following form such that the first row is the codeword whose weight is the minimum weight $d$.
$$G= \left[ \begin{array}{c|c|c} 1\ldots 1 & 1 \ldots 1\ & 0 \ldots 0 \\ \undermat{i_1}{0 \ldots 0} & \undermat{i_2}{1 \ldots  1} & \undermat{i_3}{1 \ldots 1} \end{array} \right]$$
\bigbreak

Suppose $d = \left\lfloor \frac{2(3i-1)}{3} \right\rfloor = 2i-1$, i.e., $i_1+i_2 = 2i-1$. This implies that $i_3 = i$. Note that $i_2+i_3 \geq 2i-1$ which implies $i_2 \geq i-1$ and $i_1+i_3 \geq 2i-1$ which implies $i_2 \geq i-1$. This leaves only two possible cases: $(i_1, i_2, i_3) = (i-1,i,i),(i,i-1,i)$, each of which gives $GG^T = \left[ \begin{array}{cc} 1 & 0 \\ 0 & 0 \end{array} \right], \left[ \begin{array}{cc} 1 & 1 \\ 1 & 1 \end{array} \right]$, respectively. In both cases, det$(GG^T) = 0$ and therefore they are not LCD. So there is no LCD code with minimum distance $\left\lfloor \frac{2n}{3} \right\rfloor$ for $n\equiv -1 \,\,(\text{mod }6)$.

Consider the case where $(i_1, i_2, i_3) = (i-1,i-1,i+1)$. Then $G$ generates a code of minimum distance $2i-2 = \left\lfloor \frac{2(3i-1)}{3} \right\rfloor -1$. For this case, det$(GG^T)=1$ and hence the code is LCD. 
\end {proof}

\section{$LCD\left[n,k\right]$}
\label{sec:4}

We begin with a rather straightforward lemma in order to prove Proposition 4.

\begin{lemma}

If $n \ge 8$, any $[n,n-3,d]$ binary code $C$ satisfies $d\leq2$.

\end{lemma}

\begin{proof} Let $G$ be a standard generator matrix of $C$, i.e.,  $G=\left[I_{n-3}\mid A\right]$ where $A$ is an $(n-3)\times3$ matrix.
By the Singleton bound, $d$$\leq$4.
But it is well known there is no non-trivial binary code achieving the bound.
Thus, we may say that $d\leq3$.
Suppose $d= 3$. Then each row of $A$ must have weight at least 2.
However there are only $\left(\begin{array}{c}
3\\
2
\end{array}\right)$+$\left(\begin{array}{c}
3\\
3
\end{array}\right)$= 4 possible choices for the rows of $A$.
Note that $A$ has at least 5 rows and this contradict our assumption that $d=3$.
Hence $d$ $\leq$2.

\end{proof}

Motivated by the above lemma, we can consider more general dimension $n-i$ rather than $n-3$ as follows.

\begin{proposition}
Given $i\geq2,$ $LCD[n,n-i]=2$ for all $n\geq2^{i}$.
\end{proposition}

\begin{proof} Let $G$ be a standard generator matrix of an $[n,n-i,d]$ binary code $C$, i.e.,
$G=\left[I_{n-i}\mid A\right]$ where $A$ is an $\left(n-i\right)\times i$ matrix.
By the Singleton bound, $d$$\leq i+1$. But there is no non-trivial binary code achieving the bound.
Thus, we may say that $d\leq i$.
If the minimum distance of $C$ is at least 3, each row of $A$ must have weight at least 2.
And there are $\left(\begin{array}{c}
i\\
2
\end{array}\right)+\left(\begin{array}{c}
i\\
3
\end{array}\right)+\cdots+\left(\begin{array}{c}
i\\
i
\end{array}\right)=2^{i}-i-1$ possible choices for the rows of $A$.
Thus, if $n-i$ (number of rows in $A$) $>$ $2^{i}-i-1$, then there exists a row of weight 1 in $A$.
This forces the minimum distance of $C$ to be 2, which is a contradiction.
Therefore, $d\leq2$ for all $n\geq2^{i}.$
Since this statement holds true for any linear code, it holds true for LCD codes as well,
i.e., $LCD\left[n,n-i\right]\leq2$ for all $n\geq2^{i}$.

 Next, we show that there exists an $[n,n-i,2]$ $LCD$ code for $n\geq2^{i}$.
For $i$ even,
let $G=\left[I_{n-i}\mid\undermat{i}{\mathbf{1}\mathbf{1}\cdots\mathbf{1}}\right]$, and for $i$ odd, let $G=\left[I_{n-i}\mid\undermat{i}{\mathbf{1}\mathbf{1}\cdots\mathbf{1}\mathbf{0}}\right]$ where $\mathbf{1}$
\\[0.1cm]

\noindent denotes the all one vector and $\mathbf{0}$ the all zero vector, both of which are of size $(n-i) \times 1$.
In both cases,  $GG^{T}=I_{n-i}$.
Thus, $G$ is a generator matrix for the $[n,n-i]~ LCD$ code with minimum distance 2.

 Hence $LCD[n,n-i]=2$ for all $n\geq2^{i}$.

\end{proof}

So far we have shown the exact value of $LCD[n,2]$. The relation between $LCD[n,k]$ and $LCD[n,k-1]$(or $LCD[n,k-2]$) was unknown before. Using the idea of principal submatrices, we have Proposition \ref{prsequence} below.

\
\begin{definition}
Let $A$ be a $k\times k$ matrix over a field. An $m\times m$
submatrix $P$ of $A$ is called a {\em principal submatrix} of $A$ if $P$
is obtained from $A$ by removing all rows and columns of $A$ indexed by the
same set $\{i_1,i_2,\ldots,
i_{n-m}\} \subset \{1,2,\ldots,k\}$.
\end{definition}

\begin{definition} (\cite{key-5}) Let $A$ be a $k\times k$ symmetric matrix over
a field. The {\em principal rank characteristic sequence} of $A$ (simply,
pr-sequence of $A$ or $pr(A)$) is defined as $pr(A)=r_{0}]r_{1}r_{2}\ldots r_{k}$ where for $1 \le m \le k$
\[
r_{m}=\begin{cases}
1 & \text{if \ensuremath{A} has an $m \times m$ principal submatrix of rank \ensuremath{m}}\\
0 & \text{otherwise}
\end{cases}
\]
 For convenience, define $r_{0}=1$ if and only if $A$ has a $0$ in the diagonal.
\end{definition}

\begin{proposition}\label{prsequence}$(\cite{key-5})$ For $k\geq2$ over a field with characteristic
$2$, a principal rank characteristic sequence of a $k\times k$ symmetric matrix $A$ is attainable if and
only if it has one of the following forms:
\begin{enumerate}
\begin{multicols}{3}
\item[$(i)$] $0]1\,\overline{1}\,\overline{0}$
\item[$(ii)$] $1]\overline{01}\,\overline{0}$
\item[$(iii)$] $1]1\,\overline{1}\,\overline{0}$
\end{multicols}
\end{enumerate}
where $\overline{1}=11\ldots1$ (or empty), $\overline{0}=00\ldots0$
(or empty), $\overline{01}=0101\ldots01$ (or empty).
\end{proposition}

\begin{proposition}\label{prop:LCD}
We have the following:
\begin{enumerate}
\item[$(i)$] If $k\geq3$ and $k$ is odd,
\[
LCD[n,k]\leq LCD[n,k-1]
\]
 for any $n\geq k$.
\item[$(ii)$] If $k\geq4$ and $k$ is even,
\[
LCD[n,k]\leq LCD[n,k-2]
\]
 for any $n\geq k$.
\end{enumerate}
\end{proposition}

\begin{proof}
\indent $(i)$ It suffices to show that any binary $[n,k]$ LCD code $C$ has an
$[n,k-1]$ LCD subcode for any odd $k \ge 3$. Let $G$ be a $k\times n$ generator matrix of $C$. Let $A=GG^{T}$ which is symmetric.
Then $rank(A)=k$. Since $k$ is odd, case $(ii)$ of Propostion \ref{prsequence} is
not possible. Thus the only possible cases  are $(i)$ and $(iii)$ of Proposition \ref{prsequence}, which are $0]11\ldots1$
and  $1]11\ldots1$, respectively. Hence, there exists a principal submatrix $P_{1}$
of rank $k-1$ which is obtained from $A$ by deleting some $i^{th}$
row and column of $A$ ($1\leq i\leq k)$.

Define $G_{1}$ to be a $(k-1)\times n$ matrix obtained from $G$
by deleting the $i^{th}$ row of $G$. Since $G_{1}G_{1}^{T}=P_{1}$
and $rank(P_{1})=k-1 \ne 0$, $P_{1}$ is invertible. Then the linear code
$C_{1}$ with generator matrix $G_{1}$ is LCD as well.

\medskip

$(ii)$ It suffices to show that any binary $[n,k]$ LCD code $C$ has an
$[n,k-2]$ LCD subcode for any even $k\ge 4$. Let $G$ be a $k\times n$ generator matrix of $C$ and $A=GG^{T}$.
Then $rank(A)=k$ since $C$ is LCD. By Propostion \ref{prsequence}, we have the following
three cases.
\begin{multicols}{3}
\item[$(i)$] $0]11\ldots1$
\item[$(ii)$] $1]0101\ldots01$
\item[$(iii)$] $1]11\ldots1$
\end{multicols}

So there exists a principal submatrix $P_{2}$ of rank $k-2$ which
is obtained from $A$ by deleting some $i^{th}$, $j^{th}$ rows
and columns of $A$ ($1\leq i\neq j\leq k$).

Define $G_{2}$ to be a $(k-2)\times n$ matrix obtained from $G$
by deleting the $i^{th}$ and $j^{th}$ rows of $G$. Since $G_{2}G_{2}^{T}=P_2$
and $rank(P_{2})=k-2 \ne 0$, $P_{2}$ is invertible. Then the linear $[n,k-2]$
code $C_{2}$ with generator matrix $G_{2}$ is LCD as well.

Since the minimum distance
of a code is always less than or equal to the minimum distance of a subcode,
this completes the proof of $(a)$ and $(b)$.

\end{proof}

In Table 1 we give exact values of $LCD[n,k]$ for $1 \le k \le n \le 12$.
Based on Proposition \ref{prop:LCD} and Table 1, we conjecture the following.

\medskip

{\bf{Conjecture}}
If $2 \le k \le n$, then $LCD[n,k]\leq LCD[n,k-1]$.

(Note: It suffices to show that this is true when $k$ is even.)

\begin{table}
{\begin{center}
{

\begin{tabular}{|c|c|c|c|c|c|c|c|c|c|c|c|c|}
\hline
$n/k$ & 1 & 2 & 3 & 4 & 5 & 6 & 7 & 8 & 9 & 10 & 11 & 12\tabularnewline
\hline
\hline
1 & 1 &  &  &  &  &  &  &  &  &  &  & \tabularnewline
\hline
2 & 1 & 1 &  &  &  &  &  &  &  &  &  & \tabularnewline
\hline
3 & 3 & 2 & 1 &  &  &  &  &  &  &  &  & \tabularnewline
\hline
4 & 3 & 2 & 1 & 1 &  &  &  &  &  &  &  & \tabularnewline
\hline
5 & 5 & 2 & 2 & 2 & 1 &  &  &  &  &  &  & \tabularnewline
\hline
6 & 5 & 3 & 2 & 2 & 1 & 1 &  &  &  &  &  & \tabularnewline
\hline
7 & 7 & 4 & 3 & 2 & 2 & 2 & 1 &  &  &  &  & \tabularnewline
\hline
8 & 7 & 5 & 3 & 3 & 2 & 2 & 1 & 1 &  &  &  & \tabularnewline
\hline
9 & 9 & 6 & 4 & 4 & 3 & 2 & 2 & 2 & 1 &  &  & \tabularnewline
\hline
10 & 9 & 6 & 5 & 4 & 3 & 3 & 2 & 2 & 1 & 1 &  & \tabularnewline
\hline
11 & 11 & 6 & 5 & 4 & 4 & 4 & 3 & 2 & 2 & 2 & 1 & \tabularnewline
\hline
12 & 11 & 7 & 6 & 5 & 4 & 4 & 3 & 2 & 2 & 2 & 1 & 1\tabularnewline
\hline
\end{tabular}}
\caption {$LCD[n,k]\: for\:1\leq k \le n \leq12$}
\end{center}
}
\end{table}

\section{$LCK\left[n,d\right]$}
\label{sec:5}
We define another combinatorial function $LCK\left[n,d\right]$.

\begin{definition}
$LCK[n,d]:=\max\{k~|~{\mbox{there exists a binary }} [n,k,d]~ {\mbox{LCD code}}\}$
\end{definition}
For convenience, define $LCK[n,d]=0$ if and only if  there is no LCD code with the given $n$ and $d$.

\begin{table}[H]
\begin{center}
{
\begin{tabular}{|c|c|c|c|c|c|c|c|c|c|c|c|c|}
\hline
$n/d$ & 1 & 2 & 3 & 4 & 5 & 6 & 7 & 8 & 9 & 10 & 11 & 12 \tabularnewline
\hline
\hline
1 & 1 &  &  &  &  &  &  &  &  &  &  &  \tabularnewline
\hline
2 & 2 & 0 &  &  &  &  &  &  &  &  &  &     \tabularnewline
\hline
3 & 3 & 2 & 1 &  &  &  &  &  &  &  &  &
    \tabularnewline
\hline
4 & 4 & 2 & 1$^{\ast}$ &0 &  &  &  &  &  &  &  &    \tabularnewline
\hline
5 & 5 & 4 & 1 &0  & 1 &  &  &  &  &  &  &    \tabularnewline
\hline
6 & 6 & 4 & 2 & 2 & 1$^{\ast}$ &0  &  &  &  &  &  &     \tabularnewline
\hline
7 & 7 & 6 & 3 & 2 & 1$^{\ast}$ & 0 & 1 &  &  &  &  &   \tabularnewline
\hline
8 & 8 & 6 & 4$^{\ast}$ &2  & 2 &0  & 1$^{\ast}$ &0  &  &  &  &     \tabularnewline
\hline
9 & 9 & 8 & 5 & 4 &2$^{\ast}$  & 2 &1  & 0 & 1 &  &  &      \tabularnewline
\hline
10 & 10 & 8 & 6 & 4 & 3 & 2 & 1 & 0 & 1$^{\ast}$ &0  &  &    \tabularnewline
\hline
11 & 11 & 10 & 7 & 6 & 3 & 2 & 1 & 0 &1  & 0 & 1 &     \tabularnewline
\hline
12 & 12 & 10 & 7 & 6 & 4 & 3 & 2$^{\ast}$ & 0 &1  &0  & 1$^{\ast}$ &0  \tabularnewline
\hline
\end{tabular}}
\caption {$LCK[n,d]\: for\:1\leq d \le n \leq12$}
\end{center}
\end{table}

 It is noticeable in Table 2  that more zeros  appear as $n$  gets larger. Dougherty et al. \cite{key-1} showed   $LCK[n,d]=0$ for  $n$  even and when $d=n$. Now we show that this is a special case of the following general proposition.

\begin{proposition}\label{prop:lck1}
\begin{itemize}
\item[(i)] Suppose that $n$ is even, $k\ge 1$, and $i\ge 0$. If $n\ge 6i$, then there is no $[n,k,n-2i]$ LCD code, i.e., $LCK[n,n-2i]=0$.
\item[(ii)] Suppose that $n$ is odd, $k\ge 1$, and $i\ge 0$. If $n> 6i+3$, then there is no $[n,k,n-2i-1]$  LCD code, i.e., $LCK[n,n-2i-1]=0$.
\end{itemize}
\end{proposition}

\begin{proof}
\indent $(i)$ Suppose $C$ is an LCD $[n,k,n-2i]$ code with parameters in the hypothesis. Let $G$ be a generator matrix of $C$.

If $k=1$, then $GG^T=0$ since the minimum distance $n-2i$ is even. Then by Proposition \ref{prop:1}, there is no $[n,1,n-2i]$ LCD code with $n$ even.

Now suppose $k\ge 2$. Then there should exist an LCD $[n,2,n-2i]$ subcode of $C$. By the Griesmer Bound with $k=2$, we obtain $n\ge n-2i+\frac{n-2i}{2}$ which implies $n \le 6i$. Thus  we can say that  there is no $[n,2,n-2i]$ code if $n>6i$. When $n$ meets the Griesmer Bound, i.e., $n=6i$, there is no $[6i,2,4i]$ LCD code because by Proposition \ref{prop:3} the maximum of the possible minimum distance among any $[6i,2]$ LCD codes is $4i-1$.

\indent $(ii)$ A similar argument to $(i)$ shows that there is no $[n,1,n-2i-1]$ LCD code with $n$ odd because the minimum distance $n-2i-1$ is even.

 Suppose $k\ge 2$. Then there should exist an LCD $[n,2,n-2i-1]$ subcode of $C$. By the Griesmer Bound with $k=2$, we have
  $n\ge n-2i-1+\frac{n-2i-1}{2}$ which implies $n \le 6i+3$. Thus  we can say that  there is no $[n,2,n-2i-1]$ code if $n>6i+3$.   That is, there is no such an LCD code.
\end{proof}

In Table 2,  the values of $LCK\left[n,d\right]$ are given for $1 \le d \le n \le 12$ . These values are obtained using Table 1, Proposition~\ref{prop:lck1}, and two tables from \cite{key-1}. The values with $^\ast$ are the ones that are  corrected here as they are  incorrectly  reported in Table 1 of \cite{key-1}.

\section{Appendix}

Below is an exhaustive search program written by MAGMA~\cite{Mag} in order to compute  $LCD[n,k]$ which run slowly for large $n$ and $k$.

\begin{verbatim}
LCD:=function(n,k)
I:=IdentityMatrix(GF(2),k);

Max:=0;
for g in RMatrixSpace(GF(2),k,n-k) do
if Determinant(I+(g*Transpose(g))) eq 1
then if Max lt MinimumDistance
(LinearCode(HorizontalJoin(I,g)))
then Max:=MinimumDistance
(LinearCode(HorizontalJoin(I,g)));
end if;
end if;
end for;
return Max;
end function;
\end{verbatim}



\begin{thebibliography}{}
\bibitem{key-5}Barret, W., Butler, S., Catral, M.,  Fallat,  S.M.,  Hall, H.T.,  Hogben,  L.,   Driessche,  P.van den, Young, M.: The principal rank characteristic sequence over various fields. Linear Algebra and its Applications. 459, 222-236 (2014)

\bibitem{Mag} Bosma, W., Cannon, J.: Handbook of Magma functions, Sydney (1995)

\bibitem{C1}Carlet, C.,  Guilley, S.:  Complementary dual codes for counter-measures to side-channel attacks. In Coding Theory and Applications.  Springer International Publishing, 97-105 (2015)



\bibitem{D1} Ding, C., Li, C.,  Li, S.:  LCD cyclic codes over finite fields. arXiv preprint arXiv:1608.02170 (2016)

\bibitem{key-1} Dougherty, S.T,  Kim, J.-L.,  Ozkaya, B.,
Sok, L.,  Sole, P.: The combinatorics of LCD codes : Linear Programming
bound and orthogonal matrices. Submitted to Linear Algebra and Applications
on June, 1 (2015)

\bibitem{E1}Esmaeili, M.,  S. Yari.: On complementary-dual quasi-cyclic codes. Finite Fields and Their Applications 15, no. 3,  375-386 (2009)

\bibitem{M1} Massey, J. L.:  Reversible codes. Information and Control, 7(3), 369-380  (1964)

\bibitem{M2} Massey, J. L.: Linear Codes with Complementary Duals. Discrete Mathematics. 106-107, 337-342 (1992)


\bibitem{M4} Mesnager, S., Tang, C.,  Qi, Y.: Complementary Dual Algebraic Geometry Codes. arXiv preprint arXiv:1609.05649 (2016)

\bibitem{M3} Muttoo, S.K.,  Lal, S.:  A reversible code over $ GF (q) $. Kybernetika, 22(1),  85-91 (1986)

\bibitem{S1} Sendrier, N.:  Linear codes with complementary duals meet the Gilbert-Varshamov bound. Discrete mathematics. 285(1), 345-347 (2004)

\bibitem{T1} Tzeng, K.,  Hartmann, C.:  On the minimum distance of certain reversible cyclic codes (Corresp.). IEEE Transactions on Information Theory. 16(5), 644-646 (1970)

\bibitem{V1} Kandasamy, W.V., Smarandache, F., Sujatha, R.,  Duray, R.R.:  Erasure Techniques in MRD codes. Infinite Study (2012)

\bibitem{key-6} Huffman, W.C.,  Pless, V.:   Fundamentals of error-correcting codes. Cambridge university press (2010)

\bibitem{Y1} Yang, X.,  Massey, J.L.: The condition for a cyclic code to have a complementary dual. Discrete Mathematics. 126(1-3),  391-393 (1994)

\end{thebibliography}


\end{document}